\documentclass[11pt,twoside]{article}

\usepackage{a4wide}
\usepackage[english]{babel}
\usepackage{amssymb, amsmath, amsthm, amsfonts}
\usepackage{url}
\usepackage{natbib}
\bibpunct{[}{]}{,}{n}{,}{,}

\usepackage[T1]{fontenc}
\newcommand{\DD}{{\rm d}}
\newcommand{\I}{{\rm i}}
\newcommand{\dy}{{{\rm d}}y}
\newcommand{\dx}{{{\rm d}}x}
\newcommand{\Log}{{\rm Log}}

\newcommand{\C}{ \mathbb{C} }

\newcommand{\N}{ \mathbb{N} }

\newcommand{\pd}[2]{ \frac{\partial #1}{\partial #2} }

\newcommand{\R}{ \mathbb{R} }

\newcommand{\Var}{{\mathbb V}{\rm ar}\,}
\newcommand{\ra}{\rightarrow}
\newcommand{\prob}[1]{\mathbb{P}\left(#1\right)}

\newcommand{\E}{\mathbb{E}}
\newcommand{\da}{\downarrow}

\usepackage{parskip}
\theoremstyle{plain}
\newtheorem{theorem}{Theorem}     
\newtheorem{lemma}{Lemma}

\theoremstyle{definition}

\newtheorem{remark}{Remark}

\begin{document}

\title{Convergence rates of Laplace-transform based estimators}
\author{Arnoud V. den Boer$^{1,2}$\footnote{Corresponding author. Email: {\tt {\scriptsize a.v.denboer@utwente.nl}}.}, \:\:Michel Mandjes$^{2,3,4}$ \\ 
\small$^1$ University of Twente, Drienerlolaan 5, 7522 NB Enschede \\
\small$^2$ Centrum Wiskunde \& Informatica, Science Park 123, 1098 XG Amsterdam\\
\small$^3$ University of Amsterdam, Science Park 904, 1098 XH Amsterdam\\
\small$^4$ EURANDOM, P.O. Box 513, 5600 MB Eindhoven
}

\maketitle

\begin{abstract} This paper considers the problem of estimating probabilities of the form $\mathbb{P}(Y \leq w)$, for a given value of $w$, in the situation that a sample of i.i.d.\ observations $X_1, \ldots, X_n$ of $X$ is available, and where we explicitly know a functional relation between the   Laplace transforms of the non-negative random variables $X$ and $Y$.
A plug-in estimator is constructed by calculating the Laplace transform of the empirical distribution of the sample $X_1, \ldots, X_n$, applying the functional relation to it, and then (if possible) inverting the resulting Laplace transform and evaluating it in $w$. We show, under mild regularity conditions, that the resulting estimator is weakly consistent and has expected absolute estimation error $O(n^{-1/2} \log(n+1))$. We illustrate our results by two examples: in the first we estimate the distribution of the workload in an M/G/1 queue from observations of the input in fixed time intervals, and in the second we identify the distribution of the increments when observing a compound Poisson process at equidistant points in time (usually referred to as `decompounding').
\end{abstract}

\section{Introduction} \label{sec:intro}
The estimation problem considered in this paper is the following. 
Suppose we have independent observations of the (nonnegative) random variable $X$, but we are interested in estimating the distribution of the (nonnegative) random variable $Y$. The crucial element in our set up is that we explicitly know the relation between the Laplace transforms of the random variables $X$ and $Y$, i.e., we have a mapping $\Psi$ which maps Laplace transforms of random variables to complex-valued functions defined on the right-half complex plane, and which maps the Laplace transform of $X$ to the Laplace transform of $Y$. 

A straightforward estimation procedure could be the following. (i)~Estimate the Laplace transform of $X$ by its evident empirical estimator; denote this estimate by $\tilde{X}_n$; (ii)~estimate the Laplace transform of $Y$ by $\Psi \tilde{X}_n$; 
 (iii)~apply Laplace inversion on $\Psi \tilde{X}_n$, so as to obtain an estimate of the distribution of $Y$. To justify this procedure, there are several issues that need to be addressed. First, $\tilde{X}_n$ may not lie in the domain of the mapping $\Psi$, and second, $\Psi \tilde{X}_n$ may not be a Laplace transform, and thus not amenable for Laplace inversion.

The main contribution of this paper is that we specify a procedure in which the above caveats are addressed, leading to the result that the plug-in estimator described above converges, in probability, to the true value as $n$ grows large. In addition we have bounds on its performance: the
expected absolute estimation error is $O(n^{-1/2} \log(n+1))$. 
Perhaps surprisingly, the techniques used primarily rely on an appropriate combination of standard proof techniques. Our result is valid under three mild regularity conditions: two of them are essentially of a technical nature, whereas the third can be seen as a specific continuity property that needs to be imposed on the mapping $\Psi$.

In this paper, two specific examples are treated in greater detail. In the first, an M/G/1 queueing system is considered: jobs of random size arrive according to a Poisson process with rate $\lambda>0$, the job sizes are i.i.d.\ samples from a nonnegative random variable $B$, and the system is drained at unit rate. Suppose that we observe the amount of work arriving in intervals of fixed length, say $\delta>0$; these observations are compound Poisson random variables, distributed as
\[X\stackrel{\rm d}{=} \sum_{i=1}^N B_i,\] with $N$ Poisson distributed with mean $\lambda\delta$, independent of the job sizes, and with $B_1, B_2, \ldots$ mutually independent and distributed as $B$. We show how our procedure can be used to estimate the  distribution of the workload $Y$ from the compound Poisson observations; the function $\Psi$ follows from the Pollaczek-Khinchine formula. As we demonstrate, the regularity conditions mentioned above are met.
In the second example, often referred to as `decompounding', the goal is to determine the job size distribution from compound Poisson observations. 

\vspace{2mm}

{\it Literature.} Related work can be found in various branches of the literature. Without aiming at giving an exhaustive overview, we discuss some of the relevant papers here.
The first branch consists of papers on estimating the probability distribution of a non-observed random variable by exploiting a given functional relation between the Laplace transforms of $X$ and $Y$.
The main difficulty that these papers circumvent is the issue of `ill-posedness': a sequence of functions $(f_n)_{n \in \N}$ may not converge to a function $f$, as $n$ grows large, even if the corresponding Laplace transforms of $f_n$ do converge to the Laplace transform of $f$.  
Remedies, based on `regularized Laplace inversion' have been proposed, in a compound Poisson context, by Shimizu \cite{Shimizu2010} (including Gaussian terms as well) and Mnatsakanov {\it et al.} \cite{MnatsakanovRuymgaartRuymgaart2008}; the rate of convergence is typically just $1/\log n$ in an appropriately chosen $L_2$-norm. Hansen and Pitts \cite{HansenPitts2006} use the Pollaczek-Khinthcine formula to construct estimators for the service-time distribution and its stationary excess distribution in an $M/G/1$ queue, and show that the estimated stationary excess distribution is asymptotically Normal. 

Some related papers that use Fourier instead of Laplace inversion are \cite{VanEsEtAl2007}, \cite{ComteDuvalGenonCatalot2014}, \cite{ComteDuvalGenonCatalotKappus2014} and \cite{HallPark2004}. Van Es {\it et al.} \cite{VanEsEtAl2007} estimate the density of $B_i$ by inverting the empirical Fourier transform associated with a sample of $X$, and prove that this estimator is weakly consistent and asymptotically normal.
Comte {\it et al.} \cite{ComteDuvalGenonCatalot2014} also estimate the density of $B_i$ using the empirical Fourier transform of $X$, by exploiting an explicit relation derived by Duval \cite{Duval2013} between the density of $X$ and $B_i$. They show that this estimator achieves the optimal convergence rate in the minimax sense over Sobolev balls. Comte {\it et al.} \cite{ComteDuvalGenonCatalotKappus2014} extend this to the case of mixed compound Poisson distributions (where the intenstiy $\lambda$ of the Poisson process is itself a random variable), and provide bounds on the $L^2$-norm of the density estimator. 
Finally, Hall and Park \cite{HallPark2004} estimate service-time characteristics from busy period data in an infinite-server queueing setting, and prove convergence rates in probability.

A second branch of research concerns methods that do not involve working with transforms and inversion. Buchmann and Gr\"ubel \cite{BuchmannGruebel2003} develop a method for decompounding: in the case the underlying random variables have a discrete distribution by relying on the so-called Panjer recursion, and in the case of continuous random variables by expressing the distribution function of the summands $B_i$ in terms of a series of alternating terms involving convolutions of the distribution of $X$. The main result of this paper concerns the asymptotic Normality of specific plug-in estimators. This method having the inherent difficulty that probabilities are not necessarily estimated by positive numbers, an advanced version (for the discrete case only) has been proposed by the same authors in \cite{BuchmannGruebel2004}. This method has been further extended by B{\o}gsted and Pitts \cite{BogstedPitts2010} to that of a general (but known) distribution for the number of terms $N$.
Duval \cite{Duval2013} estimates the probability density of $B_i$ by exploiting an explicit relation between the densities of $X$ and $B_i$, which however is only valid if $\lambda \delta < \log 2$. She shows that minimax optimal convergence rates are achieved in an asymptotic regime where the sampling rate $\delta$ converges to zero.
The introduction of \cite{BogstedPitts2010} gives a compact description of the state-of-the-art of this branch of the literature.

A third body of work concentrates on the specific domain of queueing models, and develops techniques to efficiently estimate large deviation probabilities. Bearing in mind that estimating small tail probabilities directly from the observations may be inherently slow and inefficient \cite{GlynnTorres1996}, techniques have been developed that exploit some structural understanding of the system. Assuming exponential decay in the exceedance level, the pioneering work of Courcoubetis {\it et al.} \cite{CourcoubetisEtAl1995}  provide (experimental) backing for an extrapolation technique. The approach proposed by Zeevi and Glynn  \cite{ZeeviGlynn2004} has provable convergence properties; importantly, their results are valid in great generality, in that they cover e.g.\ exponentially decaying as well as Pareto-type  tail probabilities. Mandjes and van de Meent  \cite{MandjesVandeMeent2009} consider queues with Gaussian input; it is shown how to accurately estimate the characteristics of the input stream by just measuring the buffer occupancy; interestingly, and perhaps counter-intuitively, relatively crude periodic measurements are sufficient to estimate fine time-scale traffic characteristics. 

As it is increasingly recognized that probing techniques may play a pivotal role when designing distributed control algorithms, there is a substantial number of research papers focusing on applications in communication networks. A few examples are the procedure of Baccelli {\it et al.}\ \cite{BaccelliKauffmannVeitch2009} that infers input characteristics from delay measurements, and  the technique of Antunes and Pipiras \cite{AntunesPipiras2011} that estimates the interrenewal distribution based on probing information.
This paper contributes to this line of research by showing how a Laplace-transform based estimator, using samples of the workload obtained by probing, can be used to estimate the workload in an M/G/1 queue; cf.\ Section \ref{subsec:mg1} and \ref{sec:num}.

\vspace{2mm}

{\it Organization.}
The rest of this paper is organized as follows. In Section \ref{sec:method} we formally define our Laplace-transform based estimator, and Section \ref{sec:rates} shows that the expected absolute estimation error is $O(n^{-1/2} \log(n+1))$. In Section \ref{subsec:mg1} we apply this result to an estimation problem in queueing theory, and in Section \ref{subsec:decompounding} to a decompounding problem. Section \ref{sec:aux}  contains a number of auxiliary lemmas used to prove the main theorems in this paper. A numerical illustration is provided in Section \ref{sec:num}.

\vspace{2mm}

{\it Notation.} We finish this introduction by introducing  notation that is used throughout this  paper.
The real and imaginary part of a complex number $z \in \C$ are denoted by $\Re(z)$ and $\Im(z)$; we use the symbol $\I$ for the imaginary unit. We write $\C_+ := \{ z \in \C \mid \Re(z) \geq 0 \}$ and $\C_{++} := \{ z \in \C \mid \Re(z) > 0 \}$. For a function $f: [0, \infty) \ra \R$, let $\bar{f}(s) = \int_0^{\infty} f(x) e^{-s x} \dx$ denote the Laplace transform of $f$, defined for all $s \in \C$ where the integral is well-defined.
For any nonnegative random variable $X$, let $\tilde{X}(s) := \E[\exp(-s X)]$ denote the Laplace transform of $X$, defined for all $s \in \C_+$. (Although $\tilde{X}(s)$ may be well-defined for $s$ with $\Re(s) < 0$, we restrict ourselves without loss of generality to $\C_{+}$, which is contained in the domain of $\tilde{X}(s)$ for each nonnegative random variable $X$.) For $t \in (0, \infty)$, as usual, $\Gamma(t) := \int_0^{\infty} x^{t-1} e^{-x} {\rm d}x$ denotes the Gamma function. The complement of an event $A$ is written as $A^{\mathfrak{c}}$; the indicator function corresponding to $A$ is given by ${\bf 1}_A$.  

\section{Laplace-transform based estimator} \label{sec:method}
In this section we formally define our plug-in estimator. The setting is as sketched in the introduction: we have $n$ i.i.d.\ observations $X_1, \ldots, X_n$ of the random variable $X$ at our disposal, and we wish to estimate the distribution of $Y$, where we know a functional relation between the transforms of $X$ and $Y$.

Let $\mathcal{X}$ be a collection of (single-dimensional) nonnegative random variables, and let the collection
$\tilde{\mathcal{X}} = \{ \tilde{X}(\cdot) \mid X \in \mathcal{X} \}$ represent their Laplace transforms.
Let \[ \Psi: \tilde{\mathcal{X}} \ra \{ g: \C_{+} \ra \C  \} \]
map each Laplace transform in $\tilde{\mathcal{X}}$ to a complex-valued function on $\C_{+}$.
Finally, let $Y$ be a nonnegative random variable such that $\tilde{Y}(s) = (\Psi \tilde{X})(s)$ for some unknown $X \in \mathcal{X}$ and all $s \in \C_{+}$, i.e., $\Psi$ maps the Laplace transform of $X$ onto the Laplace transform of $Y$. 
  We are interested in estimating the cumulative distribution function $F^Y(w)$ of $Y$ at a given value $w > 0$, based on the sample $X_1, \ldots, X_n$. 
  The distributions of both $X$ and $Y$ are assumed to be unknown, but the mapping $\Psi$ {\it is} known (and will be exploited in our estimation procedure).

A natural approach to this estimation problem is to (i) estimate the Laplace transform of $Y$ by
$\Psi \tilde{X}_n$, where $\tilde{X}_n$ is the `na\"{\i}ve' estimator
\[ \tilde{X}_n(s) = \frac{1}{n} \sum_{i=1}^n \exp(-s X_i), \quad (s \in \C_{+}); \]
observe that $\tilde{X}_n$ can be interpreted as the Laplace transform of the empirical distribution of the sample $X_1, \ldots, X_n$; then (ii) estimate the Laplace transform corresponding to the distribution function $F^Y$ by $s \mapsto s^{-1} (\Psi \tilde{X}_n)(s)$; and finally (iii) apply an inversion formula for Laplace transforms and evaluate the resulting expression in $w$. Note that in step (ii) we relied on the standard identity
\[\int_0^\infty e^{-sw} F^Y(w)\,{\rm d}w = \frac{\E[e^{-sY}]}{s}.\]

There are two caveats, however, with this approach: first, the transform $\tilde{X}_n$ is not necessarily an element of $\tilde{\mathcal{X}}$, in which case $\Psi \tilde{X}_n$ is undefined, and second, the function $s^{-1} (\Psi \tilde{X}_n)(s)$ is not necessarily a Laplace transform and thus not amenable for inversion.

To overcome the first issue, we let $E_n$ be the event that $\tilde{X}_n \in \tilde{\mathcal{X}}$. We assume that $E_n$ lies in the natural filtration generated by $X_1, \ldots, X_n$, is not defined in terms of characteristics of the (unknown) $X$, and also that $\prob{E_n^{\mathfrak{c}}} \ra 0$ as $n \ra \infty$. 
For the main result of this paper, Theorem \ref{thm:convergenceRate} in Section \ref{sec:rates}, it turns out to be irrelevant how $F^Y(w)$ is estimated on $E_n^{\mathfrak{c}}$ (as long as the estimate lies in $[0,1]$); we could, for example, estimate it by zero on this event. In concrete situations, it is typically easy to determine a suitable choice for the sets $E_n$; for both applications considered in Section \ref{sec:appl}, we explicitly identify the $E_n$.

On the event $E_n$, we estimate the Laplace transform of $F^Y$ by the plug-in estimator
\begin{align} \label{eq:emplap}
\bar{F}^Y_n(s) := \frac{1}{s} (\Psi \tilde{X}_n)(s), \quad (s \in \C_{++}).
\end{align}
To overcome the second issue, of $\bar{F}^Y_n(s)$ not necessarily being a Laplace transform, we estimate $F^Y(w)$ by applying a \emph{truncated} version of Bromwich' Inversion formula \citep{Doetsch1974}: 
\begin{align} \label{eq:estimator}
F^Y_n(w) = \int_{-\sqrt{n}}^{\sqrt{n}} \frac{1}{2 \pi} e^{(c + \I y) w} \bar{F}^Y_n(c + \I y) \dy,
\end{align}
where $c$ is an arbitrary positive real number. In the `untruncated' version of Bromwich' Inversion formula the integration in \eqref{eq:estimator} is over the whole real line. Since that integral may not be well-defined if $\bar{F}^Y_n$ is not a Laplace transform, we integrate over a finite interval (which grows in the sample size $n$).

The thus constructed estimator has remedied the two complications that we identified above. The main result of this paper, which describes the performance of this estimator as a function of the sample size $n$, is given in the next section.

\section{Main result: convergence rate} \label{sec:rates}
In this section we show that the expected absolute estimation error of our estimator $F^Y_n(w)$, as defined in the previous section, is bounded from above by
a constant times $n^{-1/2} \log(n+1)$. 

Our result is proven under the following assumptions. 

\begin{itemize}
\item[(A1)]
For each $n \in \N$ there is an event $A_n \subset E_n$, such that
$\prob{A_n^{\mathfrak{c}}} \leq \kappa_1 n^{-1/2}$ for some $\kappa_1 > 0$ independent of $n$;
\item[(A2)] $F^Y(y)$ is continuously differentiable on $[0, \infty)$, and twice differentiable at $y=w$;
\item[(A3)] There are constants $\kappa_2 \geq 0$, $\kappa_3 \geq 0$ and (nonnegative and random) $Z_n$, $n \in \N$,
such that $\sup_{p \in (1,2)} \E[|Z_n|^p] \leq \kappa_3 n^{-1/2}$  for all $n \in \N$, and such that, on the event $A_n$,
\[ | (\Psi \tilde{X}_n)(s) - (\Psi \tilde{X})(s) | \,\leq \,\kappa_2 \Big|\tilde{X}_n(s) - \tilde{X}(s) \Big| +Z_n \:\:\text{ a.s.},\]
for all $s  = c + \I y$, $n \in \N$, and $-\sqrt{n} \leq y \leq \sqrt{n}$.
\end{itemize}
These assumptions are typically `mild'; we proceed with a short discussion of each of them.

Assumption (A1) ensures that the contribution of the complement of $A_n$ (and also that of the complement of $E_n$) to the expected absolute estimation error is sufficiently small.
The difference between $A_n$ and $E_n$ is that the definition of $E_n$ does not involve the unknown $X \in \mathcal{X}$ (which enables us to define the estimator $F_n^Y(w)$ \emph{without} knowing the unknown $X$), whereas $A_n$ may actually depend on $X$. It is noted that in specific applications, this helps when checking whether the assumptions (A1)--(A3) are satisfied; cf.\ the proofs of Theorems~\ref{thm:mg1rates} and \ref{thm:decompoundingrates}. If $\prob{A_n^{\mathfrak{c}}} \sim n^{-a}$ for some $a \in (-1/2,0)$ then (A1) does not hold and Theorem \ref{thm:convergenceRate} is not valid.

Assumption (A2) is a smoothness condition on $F^Y$ that we use to control the error caused by integrating in
\eqref{eq:estimator} over a finite interval, instead of integrating over $\R$. The twice-differentiability assumption is only
used to apply Lemma \ref{lem:ntaillemma} with $f = F^Y$ in the proof of Theorem \ref{thm:convergenceRate}. It can be replaced by any other condition that guarantees that, for all $n \in \N$ and some $\kappa_4 > 0$ independent of $n$,
\[ \left| \int_{|y|>\sqrt{n}} \frac{1}{2 \pi} e^{(c+\I y) w} \bar{F}^Y(c+\I y) \DD y
\right|  \leq \frac{\kappa_4}{\sqrt{n}}. \]
Continuous differentiability of $F^Y$ makes sure that
Bromwich' Inversion formula applied to $\bar{F}^Y$ yields $F^Y$ again, i.e.,
$F^Y(w) = \int_{-\infty}^{\infty} ({2 \pi} )^{-1}e^{(c + \I y) w} \bar{F}^Y(c + \I y) \dy$, cf.\ \cite[][Chapter 4]{Schiff1999}.
This equality is still true if the derivative of $F^Y(y)$ with respect to $y$ is piecewise continuous with finitely many discontinuity points, and in addition continuous at $y=w$.

Assumption (A3) can be seen as a kind of Lipschitz-continuity condition on $\Psi$ that guarantees that $\Psi \tilde{X}_n$ is `close to' $\Psi \tilde{X}$ if $\tilde{X}_n$ is `close to' $\tilde{X}$. This condition is necessary to prove the weak consistency of our estimator. The formulation with the random variables $Z_n$ allows for a more general setting than with $Z_n = 0$, and is used in both applications in Section \ref{sec:appl}.

A straightforward example that satisfies assumptions (A1)--(A3) is the case where $X \stackrel{d}{=} Y + W$, where $W$ is a known nonnegative random variable. If the cdf of $Y$ satisfies the smoothness condition (A2), then, with $\tilde{Y}(s) = (\Psi \tilde{X})(s) := \tilde{X}(s) / \tilde{W}(s)$, $A_n^{\mathfrak{c}} = E_n^{\mathfrak{c}} = \emptyset$, $Z_n = 0$ a.s., $c > 0$ arbitrary, and 
$\kappa_2 := \sup_{s = c + \textrm{i}y, - \sqrt{n} \leq y \leq \sqrt{n}} 1 / |\tilde{Z}(s)|$, it is easily seen that assumptions (A1)--(A3) are satisfied. More involved examples that satisfy the assumptions are presented in Section \ref{sec:appl}.

\vspace{\baselineskip}

\begin{theorem}\label{thm:convergenceRate}
Let $w > 0$, $c > 0$, and assume (A1)--(A3).
Then $F^Y_n(w)$ converges to $F^Y(w)$ in probability, as $n \ra \infty$, and there is a constant $C > 0$ such that, for all $n \in \N$,
\begin{align} \label{eq:convrates}
\E[|F^Y_n(w) - F^Y(w)|] \leq C n^{-1/2} \log(n+1).
\end{align}
\end{theorem}
\begin{proof}
It suffices to prove \eqref{eq:convrates}, since this implies weak consistency of $F^Y_n(w)$.
Fix $n \in \N$. 

The proof consists of three steps. In Step 1 we bound the estimation error on the event $A_n$, in Step 2 we consider the complement $A_n^{\mathfrak{c}}$, and in Step 3 we combine Step 1 and 2 to arrive at the statement of the theorem. Some of the intermediate steps in the proof rely on auxiliary results that are presented in Section \ref{sec:aux}.

{\bf Step 1.}
We show that there are positive constants $\kappa_4$ and $\kappa_5$, independent of $n$, such that, for all $n\in{\mathbb N}$ and $p \in (1,2)$,
\begin{equation}
\label{eq:ratesonA}
     \E[|F_n^Y(w) - F^Y(w)  | \cdot {\bf 1}_{A_n}] \leq  \kappa_4 n^{-1/2} + \kappa_5 (p-1)^{-1/p} n^{1/2 - 1/p}.
\end{equation}
To prove the inequality \eqref{eq:ratesonA}, consider the following
elementary upper bound:
\begin{eqnarray} \nonumber
\lefteqn{ \E\left[\left|F_n^Y(w) - F^Y(w)  \right| \cdot {\bf 1}_{A_n}\right]}\\
& =& \E\left[\left|
\int_{-\sqrt{n}}^{\sqrt{n}} \frac{1}{2 \pi} e^{(c+\I y) w} \bar{F}^{{Y}}_{{n}}(c+\I y) \dy -
\int_{-\infty}^{\infty} \frac{1}{2 \pi} e^{(c+\I y) w} \bar{F}^{Y}(c+\I y) \dy
\:\right| \cdot {\bf 1}_{A_n} \right] \nonumber \\ \nonumber
&=& \frac{1}{2\pi}\E\left[\left|
\int_{-\sqrt{n}}^{\sqrt{n}} e^{(c+\I y) w} (\bar{F}^{{Y}}_{{n}}(c+\I y) - \bar{F}^{Y}(c+\I y)) \dy
\right.\right.\\
&&\nonumber\hspace{15mm}\left.\left.- \int_{|y|>\sqrt{n}} e^{(c+\I y) w} \bar{F}^{Y}(c+\I y) \dy
\:\right|  \cdot {\bf 1}_{A_n} \right] \\ \label{eq:term1inproof}
&\leq& \E\left[\, \left|
\int_{-\sqrt{n}}^{\sqrt{n}} \frac{1}{2 \pi} e^{(c+\I y) w} (\bar{F}^{{Y}}_{{n}}(c+\I y) - \bar{F}^{Y}(c+\I y)) \dy\,
\right| \cdot {\bf 1}_{A_n} \right] \\ \label{eq:term2inproof}
&&\hspace{15mm}+\:
\left| \,\int_{|y|>\sqrt{n}} \frac{1}{2 \pi} e^{(c+\I y) w} \bar{F}^{Y}(c+\I y) \dy
\,\right|.
\end{eqnarray}

We now treat the terms \eqref{eq:term1inproof} and \eqref{eq:term2inproof} separately, starting with the latter.
By assumption (A2) and the observation
\[\int_w^{\infty} \left(\frac{\rm d}{{\rm d}y}F^Y(y + w)\right) \frac{e^{-c y}}{y} \frac{w}{e^{-c w}} \dy \le  \int_w^{\infty} \left(\frac{\rm d}{{\rm d}y}F^Y(y + w)\right)\dy = F^Y(\infty)-F^Y(2w)<1,\]
we conclude that
\[\int_w^{\infty} \left|\frac{\rm d}{{\rm d}y}F^Y(y + w)\right|\, \frac{e^{-c y}}{y} \dy < \frac{e^{-c w}}{w} < \infty,\] and therefore $F^Y$ satisfies the conditions of Lemma \ref{lem:ntaillemma}. As a result,
\eqref{eq:term2inproof} satisfies
\begin{align} \label{eq:term2binproof}
 \left| \int_{|y|>\sqrt{n}} \frac{1}{2 \pi} e^{(c+\I y) w} \bar{F}^{Y}(c+\I y) \dy
\right| \leq \frac{\kappa_4}{\sqrt{n}}, \end{align}
for some constant $\kappa_4 > 0$ independent of $n$.

We now bound the term \eqref{eq:term1inproof}. It is obviously majorized by
\[\E\left[
\int_{-\sqrt{n}}^{\sqrt{n}} \frac{1}{2 \pi} e^{c w} \left| \bar{F}^{{Y}}_{{n}}(c+\I y) - \bar{F}^{Y}(c+\I y) \right| \dy  \cdot {\bf 1}_{A_n}
 \right].\]
Let $p\in (1,2)$ and $q > 1$, with $p^{-1} + q^{-1} = 1$. By subsequent application of H\"older's Inequality,
this  expression is further bounded by
\[\E\left[
\left(\int_{-\sqrt{n}}^{\sqrt{n}} \left(\frac{1}{2 \pi} e^{c w} \right)^q \dy \right)^{1/q}
\left(\int_{-\sqrt{n}}^{\sqrt{n}} \left| \bar{F}^{{Y}}_{{n}}(c+\I y) - \bar{F}^{Y}(c+\I y) \right|^p \dy \right)^{1/p} \cdot {\bf 1}_{A_n}
 \right].\]
By computing the first integral, and an application of Jensen's inequality, this is not larger than
\[e^{c w} \frac{(2 \sqrt{n})^{1/q}}{2 \pi}
\left(\E\left[  \int_{-\sqrt{n}}^{\sqrt{n}} \left| \bar{F}^{{Y}}_{{n}}(c+\I y) - \bar{F}^{Y}(c+\I y) \right|^p \dy \cdot {\bf 1}_{A_n} \right]\right)^{1/p}
.\]
Finally applying Fubini's Theorem, we arrive at the upper bound
\begin{equation}
 \label{eq:term3inproof}
 e^{c w} \frac{(2 \sqrt{n})^{1/q}}{2 \pi}
\left(  \int_{-\sqrt{n}}^{\sqrt{n}} \E\left[ \left| \bar{F}^{{Y}}_{{n}}(c+\I y) - \bar{F}^{Y}(c+\I y) \right|^p \cdot {\bf 1}_{A_n} \right] \dy \right)^{1/p}.
 \end{equation}
We now study the behavior of (\ref{eq:term3inproof}), being an upper bound to (\ref{eq:term1inproof}),  as a function of $n$. To this end, we first derive a bound on the integrand.
Assumption (A3) implies that there exists a sequence of nonnegative random variables $Z_n$, $n\in{\mathbb N}$, such that
\begin{eqnarray} \nonumber
\left| \bar{F}^Y_n(s) - \bar{F}^{Y}(s) \right| \cdot {\bf 1}_{A_n}
&=& \left| s^{-1} (\Psi \tilde{X}_n) (s) - s^{-1} (\Psi \tilde{X})(s) \right| \cdot {\bf 1}_{A_n} \\ \label{eq:eqnwithsinverse}
&\leq & \left( \kappa_2 \left| \tilde{X}(s) - \tilde{X}_n(s) \right| + Z_n \right) \cdot |s^{-1}| \cdot  {\bf 1}_{A_n}  \text{ a.s.,}
\end{eqnarray}
for all $s  = c + \I y$ with $-\sqrt{n} \leq y \leq \sqrt{n}$. 

Now recall the so-called $c_r$-inequality
\[\E[|X+Y|^p] \leq 2^{p-1} (\E[|X|^p] + \E[|Y|^p]),\] and
the obvious inequality ${\bf 1}_{A_n}\leq 1$ a.s. As a consequence of
Lemma \ref{lem:sqrtlemma}, we thus obtain
\begin{eqnarray*}
\lefteqn{ \E\left[ \left| \bar{F}^{{Y}}_{{n}}(c+\I y) - \bar{F}^{Y}(c+\I y) \right|^p \cdot {\bf 1}_{A_n} \right] } \\
&\leq& 2^{p-1} \left( \kappa_2^p \E\left[ \left|  \tilde{X}(c+\I y) - \tilde{X}_n(c+\I y) \right|^p  \right]  +
 \E\left[|Z_n|^p\right]  \right) |c + \I y|^{-p}  \\
&\leq& 2^{p-1} (2^p \kappa_2^p  + \kappa_3) n^{-1/2} |c + \I y|^{-p}. 
\end{eqnarray*}
From \eqref{eq:term3inproof} and the straightforward inequality
\begin{eqnarray} \label{eq:integralofsinverse}
 \int_{-\infty}^{\infty} |c + \I y|^{-p} {\rm d} y
& \leq& \int_{-\infty}^{\infty} \frac{1}{(c^2+y^2)^{p/2}} {\rm d} y
 = c^{1-p} \int_{0}^{\infty} \frac{z^{-1/2}}{(1+z)^{p/2}} {\rm d} z \\ \nonumber
&=& C_0(p ):=c^{1-p} \pi^{1/2} \frac{\Gamma((p-1)/2)}{\Gamma(p/2)},
\end{eqnarray}
it follows that
\begin{eqnarray} \nonumber
\lefteqn{ \hspace{-5mm}\E\left[ \left|
\int_{-\sqrt{n}}^{\sqrt{n}} \frac{1}{2 \pi} e^{(c+\I y) w} (\bar{F}^{Y}_{{n}}(c+\I y) - \bar{F}^{Y}(c+\I y)) \dy
\right| \cdot {\bf 1}_{A_n} \right]} \\ \nonumber
&\leq & e^{c w} \frac{(2 \sqrt{n})^{1/q}}{2 \pi} \left(  2^{p-1} (2^p \kappa_2^p + \kappa_3) n^{-1/2} \int_{-\sqrt{n}}^{\sqrt{n}}
|c + \I y|^{-p} \dy \right)^{1/p}  
\\ \label{eq:term5inproof}
&\leq & C_1(p)\, n^{1/2-1/p},
 \end{eqnarray}
where
\[C_1(p) :=
e^{c w} \,\frac{2^{2-2/p}}{2 \pi}
 (2^p \kappa_2^p + \kappa_3)^{1/p} \left( C_0(p ) \right)^{1/p}.\]
It follows from $\Gamma((p-1)/2) = 2 \,\Gamma((p+1)/2) / (p-1)$ that
\[\lim_{p \da 1} (p-1)^{1/p} C_1(p)< \infty.\] 
This implies that there is a $\kappa_5 > 0$ such that
\begin{eqnarray} \label{eq:c1c3}
C_1(p) \leq \kappa_5 (p-1)^{-1/p} \quad \text{ for all $p \in (1,2)$}.
\end{eqnarray}
Upon combining the results presented in displays \eqref{eq:term1inproof}, \eqref{eq:term2inproof}, \eqref{eq:term2binproof}, \eqref{eq:term5inproof}, and \eqref{eq:c1c3}, we obtain
Inequality \eqref{eq:ratesonA}, as desired. 

{\bf Step 2.} On the complement of the event $A_n$ we have, by assumption (A1),
\begin{align} \label{eq:Acresult}
&\E\left[|F^{Y}_{n}(w) - F^Y(w)  | \cdot {\bf 1}_{A_n^{\mathfrak{c}}}\right] \leq
\prob{ A_n^{\mathfrak{c}} } 
\leq \kappa_1 n^{-1/2}.
\end{align}

{\bf Step 3.} When combining Inequalities \eqref{eq:ratesonA} and \eqref{eq:Acresult}, we obtain that
\begin{eqnarray*}
\E\left[|F^{Y}_{n}(w) - F^Y(w)  |\right] &= &\E\left[|F^{Y}_{n}(w) - F^Y(w) | \cdot {\bf 1}_{A_n}\right] + \E\left[|F^{Y}_{n}(w) - F^Y(w) | \cdot {\bf 1}_{ A_n^{\mathfrak{c}}}\right] \\
&\leq& \kappa_4 n^{-1/2} + \kappa_5 (p-1)^{-1/p} n^{1/2 - 1/p} + \kappa_1 n^{-1/2}.
\end{eqnarray*}
Now realize that we have the freedom to pick in the above inequality any $p\in(1,2)$. In particular, the choice $p=p_n: = 1 + 1/(2 \log(n+1)) \in (1,2)$ yields the bound
\begin{eqnarray*}
\E[|F^{Y}_{n}(w) - F^Y(w)  | &
\leq& \kappa_4 n^{-1/2}  + \kappa_5 (2 \log(n+1))^{1/p_n} n^{1/2 - 1/p_n} + \kappa_1 n^{-1/2} \\
&\leq& (\kappa_4 + 2 \kappa_5 e^{1/2} + \kappa_1) n^{-1/2} \log(n+1),
\end{eqnarray*}
using
\[\frac{1}{2}-\frac{1}{p_n}=-\frac{1}{2}+\frac{1}{1+2 \log(n+1)}\]
and \[n^{1/(1+2\log(n+1))} = \exp\left(\frac{\log n}{1+ 2 \log(n+1)}\right) \leq e^{1/2}.\] This finishes the proof of Thm.\ \ref{thm:convergenceRate}.
\end{proof}

\begin{remark} \label{rem:cdf}
Contrary to some of the literature mentioned in Section \ref{sec:intro} (e.g.\ \cite{MnatsakanovRuymgaartRuymgaart2008} and \cite{Shimizu2010}), we are not estimating a density but a cumulative distribution function. This difference translates into an additional $|s^{-1}|$ term in Equation \eqref{eq:eqnwithsinverse}, which enables us to bound the integral in Equation \eqref{eq:integralofsinverse}. This appears to be an crucial step in the proof of Theorem \ref{thm:convergenceRate}, because it means that the ill-posedness of the inversion problem (the fact that the inverse Laplace transform operator is not continuous) does not play a r\^ole: convergence of $\bar{F}_n^Y(\cdot)$ to $\bar{F}^Y(\cdot)$ implies convergence of $F_n^Y(w)$ to $F^Y(w)$. As a result, we do not need regularization techniques as in \cite{MnatsakanovRuymgaartRuymgaart2008} and \cite{Shimizu2010}. 
\end{remark}

\section{Applications} \label{sec:appl}
In this section we discuss two examples that have attracted a substantial amount of attention in the literature.
In both examples, the verification of the Assumptions (A1)--(A3) can be done, as we will demonstrate now.
\subsection{Workload estimation in an M/G/1 queue} \label{subsec:mg1}
In our first example we consider the so-called M/G/1 queue: jobs arrive at a service station according to a Poisson process with rate $\lambda>0$, where these jobs are i.i.d.\ samples with a service time distribution $B$; see e.g.\ see \cite{Cohen1982} for an in-depth account of the M/G/1 queue, and \cite{NazarathyPollet2012} for an annotated bibliography on inference in queueing models. Under the stability condition $\rho:= \lambda \E[B] \in (0,1)$ the queue's stationary workload is well defined. Our objective, motivated by the setup described in \cite{denBoerMandjesNunezZuraniewski2014}, is to estimate ${\mathbb P}(Y>w)$, where $Y$ is the stationary workload, and $w>0$ is a given threshold. The idea is that this estimate is based on samples of the queue's input process.

In more detail, the procedure works as follows.
By the Pollaczek-Khintchine formula \citep{Kendall1951}, the Laplace transform of the stationary workload distribution $Y$ satisfies the relation
\begin{eqnarray} \label{eq:PKformula}
 \tilde{Y}(s) = \frac{s (1 - \rho)}{s - \lambda + \lambda \tilde{B}(s)}, \quad s \in \C_{+}.
 \end{eqnarray}
For subsequent time intervals of (deterministic) length $\delta > 0$, the amount of work arriving to the queue is measured.
These observations are i.i.d.\ samples from a compound distribution $X \stackrel{\rm d}{=} \sum_{i=1}^N B_i$, with $N$ Poisson distributed with parameter $\lambda \delta$,
and the random variables $B_1, B_2, \ldots$ independent and distributed as $B$ (independent of $N$).
By Wald's equation we have $\E[X] = \delta \rho$, and a direct computation yields
$\tilde{X}(s) = \exp(-\lambda \delta + \lambda \delta \tilde{B}(s))$.
Combining this with \eqref{eq:PKformula}, we obtain the following relation between the Laplace transforms of $X$ and $Y$:
\begin{eqnarray} \label{eq:XinY}
\tilde{Y}(s) 
= \frac{s (1 - \delta^{-1} \E[X])}{s + \delta^{-1} \Log(\tilde{X}(s))}.
\end{eqnarray}
Here $\Log$ is the distinguished logarithm of $\tilde{X}(s)$ \citep{Chung2001}, which is convenient to work with in this context
\cite{VanEsEtAl2007}.
Our goal is to estimate ${\mathbb P}(Y \leq w)=F^Y(w)$, for a given $w > 0$, based on an independent sample $X_1, \ldots, X_n$.
We use the estimator $F_n^Y(w)$ defined in Section \ref{sec:method}, for an arbitrary $c > 0$, and with
\begin{itemize}
\item[(i)]  $\mathcal{X}$ the collection of all random variables $X'$ of the form $\sum_{i=1}^{N'} B'_i$ with $N'$ Poisson distributed with strictly positive mean, $\{B'_i\}_{i \in \N}$ i.i.d., independent of $N'$, nonnegative, and with $0 < \E[X'] < \delta$;
 \item[(ii)]
 the sets
 \[E_n := \{ 0 \le \frac{1}{n} \sum_{i=1}^n X_i < \delta \},\] so as to ensure that the `empirical occupation rate' of the queue,
 $\delta^{-1} n^{-1} \sum_{i=1}^n X_i$, is strictly smaller than one, 
 and that therefore the Pollaczek-Khintchine formula holds;
\item[(iii)] $\Psi$ defined through
\[(\Psi \tilde{X})(s) = \frac{s (1 + \delta^{-1} \tilde{X}'(0))}{s + \delta^{-1} \Log(\tilde{X}(s))},\:\:\:s \in \C_{+},\]
 where $\tilde{X}'(t)$ denotes the derivative of $\tilde{X}(t)$ in $t \in (0,\infty)$ and $\tilde{X}'(0) = \lim_{t \downarrow 0} \tilde{X}'(t) = -\E[X]$.
 \end{itemize}

 \vspace{3mm}

\begin{theorem} \label{thm:mg1rates} Consider the estimation procedure outlined above.
Suppose $F^Y$ is continuously differentiable, twice differentiable in $w$, and $\E[B^2] < \infty$.
Then there is a constant $C > 0$ such that
\[\E\left[|F_n^Y(w) - F^Y(w)|\right] \leq Cn^{-1/2} \log(n+1) \]
for all $n \in \N$.
\end{theorem}
\begin{proof}
Let $n \in \N$ be arbitrary, and define the events
\[ A_{n,1} := \left\{ \sup_{- \sqrt{n} \leq y \leq \sqrt{n}} \left|\frac{\tilde{X_n}(c + \I y)}{\tilde{X}(c + \I y)} - 1\right| \leq \min\left\{ \frac{1}{2},
\frac{c \delta (1 - \delta^{-1} \E[X])}{2 \log 4 } \right\}
   \right\}, \]
\[ A_{n,2} := \left\{ \delta^{-1} \left|\E[X] - \frac{1}{n} \sum_{i=1}^n X_i\right| \leq \delta^{-1} \E[X] (1 - \delta^{-1} \E[X])\right\}, \]
and $A_{n} := A_{n,1} \cap A_{n,2}$. We have $A_{n,2} \subset E_n$ (because, using $\rho = \delta^{-1} \E[X]$, the event $A_{n,2}$ implies $\delta^{-1} n^{-1} \sum_{i=1}^n X_i  \in [\rho^2, \rho (2 - \rho)] \subset (0,1)$) and thus $A_n \subset E_n$. 
We show that assumptions (A1)--(A3), as defined in  Section \ref{sec:rates}, are satisfied. To this end, we only need to show (A1) and (A3), since (A2) is assumed in the statement of the theorem. 

$\rhd$ Assumption {(A1).} Let \[\beta := \exp(-2 \lambda \delta) \min\left\{ \frac{1}{2},\frac{c \delta (1 - \delta^{-1} \E[X])}{2 \log 4 } \right\}.\]Then, for $s \in \C_{+}$,
 \begin{equation}\label{eq:inftildeX}
|\tilde{X}(s) | = |\exp(- \lambda \delta + \lambda \delta \tilde{B}(s))| \geq \exp(- \lambda \delta + \lambda \delta \Re( \tilde{B}(s)))\geq \exp(-2 \lambda \delta),
\end{equation}
which implies that
\begin{align*}
\prob{A_{n,1}^{\mathfrak{c}}} \leq \prob{ \sup_{- \sqrt{n} \leq y \leq \sqrt{n}} | \tilde{X_n}(c + \I y) - \tilde{X}(c + \I y) | >  \beta},
\end{align*}
so that Lemma \ref{lem:duvroye} then yields
\begin{align*}
\prob{A_{n,1}^{\mathfrak{c}}} < 4 \left(1 + \frac{8 \sqrt{n} \E[|X|]}{\beta} \right) e^{-n \beta^2 / 18} + \prob{\left| \frac{1}{n} \sum_{i=1}^n X_i \right| \geq \frac{4}{3} \E[|X|]}.
\end{align*}
Since $\sqrt{n} \exp(-n \beta^2 / 18) = O(n^{-1/2})$ and
\[\prob{\left| \frac{1}{n} \sum_{i=1}^n X_i \right| \geq \frac{4}{3} \E[|X|]} \leq
\prob{\left| \E[X] - \frac{1}{n} \sum_{i=1}^n X_i \right| \geq \frac{1}{3} \E[X] } \leq
n^{-1} \frac{9 \E[(X - \E[X])^2]}{\E[X]^2}, \]
using the nonnegativity of $X$ and Chebyshev's inequality, it follows that ${\mathbb P}({A_{n,1}^{\mathfrak{c}}} )= O(n^{-1/2})$.
It also follows easily from Chebyshev's inequality that $\prob{A_{n,2}^{\mathfrak{c}}} = O(n^{-1/2})$. It follows immediately that $\prob{A_{n}^{\mathfrak{c}}} \leq \prob{A_{n,1}^{\mathfrak{c}}} +\prob{A_{n,2}^{\mathfrak{c}}} =O(n^{-1/2})$, which implies that assumption (A1) is satisfied.

$\rhd$ {Assumption (A3).} Fix $y \in[-\sqrt{n}, \sqrt{n}]$ and $s = c + \I y$. Then
\begin{eqnarray} \nonumber
\lefteqn{ |\,(\Psi \tilde{X}_n)(s) - (\Psi \tilde{X})(s)\,|}\\
&\leq& \left|\frac{s (1 + \delta^{-1} \tilde{X}_n'(0))}{s + \delta^{-1} \Log(\tilde{X}_n(s))} - \frac{s (1 + \delta^{-1} \tilde{X}_n'(0))}{s + \delta^{-1} \Log(\tilde{X}(s))} \right| \nonumber \\ \nonumber
&&+\,  \left|\frac{s (1 + \delta^{-1} \tilde{X}_n'(0))}{s + \delta^{-1} \Log(\tilde{X}(s))} - \frac{s (1 + \delta^{-1} \tilde{X}'(0))}{s + \delta^{-1} \Log(\tilde{X}(s))} \right| \\ \nonumber
&\leq &
\left| \frac{ \delta^{-1} \Log(\tilde{X}(s)) - \delta^{-1} \Log(\tilde{X}_n(s))}{s (1 + \delta^{-1} \tilde{X}'(0))}
\cdot \frac{s (1 + \delta^{-1} \tilde{X}_n'(0))}{s + \delta^{-1} \Log(\tilde{X}_n(s))} \cdot \frac{s (1 + \delta^{-1} \tilde{X}'(0))}{s + \delta^{-1} \Log(\tilde{X}(s))} \right| \\ \nonumber
&&+ \, \left|
\frac{\delta^{-1} (\tilde{X}_n'(0)) -  \tilde{X}'(0))}{(1 + \delta^{-1} \tilde{X}'(0))} \cdot
\frac{s (1 + \delta^{-1} \tilde{X}'(0))}{s + \delta^{-1} \Log(\tilde{X}(s))} \right| \\ \nonumber
&\leq&  \frac{\delta^{-1}}{1 + \delta^{-1} \E[X]} \left| \Log(\tilde{X}(s)) - \Log(\tilde{X}_n(s)) \right|
\cdot \left| s^{-1} (\Psi \tilde{X}_n)(s) \right| \cdot \left| (\Psi \tilde{X})(s) \right|\\ \label{eq:mg1logeq1}
&&+\, \frac{\delta^{-1}}{1 + \delta^{-1} \E[X]}  \left| \E[X] -\frac{1}{n} \sum_{i=1}^n X_i \right|
\cdot \left| (\Psi \tilde{X})(s) \right|.
\end{eqnarray}

If $f: \R \ra \R$ is a continuous function with $f(0) = 1$ and $f(t) \neq 0$ for all $t \in \R$,
then for all $t$ such that $|f(t) -1| \leq \frac{1}{2}$ we have
$\Log(f(t)) = L(f(t))$, where, for $z \in \C,$ $|z-1|<1$,
\begin{align} \label{eq:Lfunction}
L(z) = \sum_{j \geq 1} \frac{(-1)^{j-1}}{j} (z-1)^j;
\end{align}
this follows from the construction of the distinguished logarithm \citep{Chung2001}. In addition, if $|z-1| \leq \frac{1}{2}$, then
\begin{align} \label{eq:Lfunction2}
|L(z)| \leq \sum_{j \geq 1} \frac{1}{j} |z-1|^j = \log\left(\frac{1}{1 - |z-1|}\right) \leq |z-1| \log 4.
\end{align}
This implies that, on $A_n$, we have 
\begin{eqnarray} \nonumber
\lefteqn{\hspace{-1cm}\left|\,\Log(\tilde{X}_n(c + \I y)) - \Log(\tilde{X}(c + \I y)) \right|
= \left|\,\Log\left( \frac{\tilde{X}_n(c + \I y)}{\tilde{X}(c + \I y)}\right) \right|
= \left|\,L\left( \frac{\tilde{X}_n(c + \I y)}{\tilde{X}(c + \I y)}\right) \right|}\\
&\leq& \left| \frac{\tilde{X}_n(c + \I y)}{\tilde{X}(c + \I y)} - 1 \right| \log 4
\leq \left| \tilde{X}_n(c + \I y) - \tilde{X}(c + \I y) \right| (\log 4 )\exp(2 \lambda \delta),\label{eq:mg1logeq2}
\end{eqnarray}
where the last inequality follows from \eqref{eq:inftildeX}.

Furthermore, we have on $A_n$ that
\begin{eqnarray} \nonumber\lefteqn{
 |s^{-1} (\Psi \tilde{X}_n)(s)| = \left| \frac{1 - \delta^{-1} \frac{1}{n} \sum_{i=1}^n X_i}{s + \delta^{-1} \Log(\tilde{X}_n(s))} \right|} \\ \nonumber
&\leq& \left| \frac{1 - \delta^{-1} \frac{1}{n} \sum_{i=1}^n X_i}{1-\delta^{-1} \E[X]} \right| \cdot
\left| \frac{s + \delta^{-1} \Log(\tilde{X}_n(s))}{s + \delta^{-1} \Log(\tilde{X}(s))} \right|^{-1}
\cdot \left| \frac{1 - \delta^{-1} \E[X]}{s + \delta^{-1} \Log(\tilde{X}(s))} \right| \\ \nonumber
&\leq& \left| 1 + \delta^{-1} \frac{\E[X] - \frac{1}{n} \sum_{i=1}^n X_i}{1-\delta^{-1} \E[X]} \right| \cdot
\left| 1 + \frac{\delta^{-1} \Log(\tilde{X}_n(s) / \tilde{X}(s))}{s + \delta^{-1} \Log(\tilde{X}(s))} \right|^{-1}
\cdot \left| s^{-1} \tilde{Y}(s) \right| \\ \nonumber
&\leq& \left(1 + \frac{\delta^{-1} \E[X] (1 - \delta^{-1} \E[X])}{1-\delta^{-1} \E[X]} \right) \cdot
\left| 1 + (\Psi \tilde{X})(s) \frac{\delta^{-1} L(\tilde{X}_n(s) / \tilde{X}(s))}{s (1 - \delta^{-1} \E[X])} \right|^{-1}
\cdot c^{-1} \\
&\leq& \left(1 + \frac{\delta^{-1} \E[X] (1 - \delta^{-1} \E[X])}{1-\delta^{-1} \E[X]} \right) \cdot 2 \cdot c^{-1}, \label{eq:mg1logeq3}
\end{eqnarray}
since $|1+z|^{-1} \leq (1 -|z|)^{-1} \leq (1 - 1/2)^{-1}$ for all $z \in \C$ with $|z| \leq \frac{1}{2}$; bear in mind that, in particular,
\begin{align*}
 \left| (\Psi \tilde{X})(s) \frac{\delta^{-1} L(\tilde{X}_n(s) / \tilde{X}(s))}{s (1 - \delta^{-1} \E[X])} \right|
\leq \frac{\delta^{-1} c^{-1} |L(\tilde{X}_n(s) / \tilde{X}(s))| }{1 - \delta^{-1} \E[X]}
\leq \frac{\delta^{-1} c^{-1} \log 4}{1 - \delta^{-1} \E[X]} \left| \frac{\tilde{X}_n(s)}{\tilde{X}(s)} - 1 \right| 
\leq \frac{1}{2}
\end{align*}
on $A_{n,1}$. Finally, writing
\[ Z_n = \frac{\delta^{-1}}{1 + \delta^{-1} \E[X]}  \left| \,\E[X] -\frac{1}{n} \sum_{i=1}^n X_i \,\right|
\cdot \left| (\Psi \tilde{X})(s) \right|, \]
and noting that $\E[X^2]< \infty$, it follows from Lemma \ref{lem:sqrtlemma} and $|(\Psi \tilde{X})(s)| \leq 1$ that $\E[|Z_n|^p] \leq \kappa_3 n^{-1/2}$ for all $p \in (1,2)$ and some $\kappa_3 > 0$ independent of $n$ and $p$. Combining this with equations \eqref{eq:mg1logeq1}, \eqref{eq:mg1logeq2}, and \eqref{eq:mg1logeq3}, implies that assumption (A3) holds.
\end{proof}

\vspace{\baselineskip}
\begin{remark}
An important problem in \cite{denBoerMandjesNunezZuraniewski2014} is to develop heuristics for choosing $\delta$, in order to minimize the expected estimation error. In the proof of Theorem \ref{thm:convergenceRate} we show the following upper bound
\begin{align} \label{eq:ubmg1}
 \E\left[ | F_n^Y(w) - F^Y(w)|\right] &\leq
  \kappa_4 n^{-1/2} + C_1(p) n^{1/2 - 1/p} + \prob{A_n^{\mathfrak{c}}},
 \end{align}
where $p =p_n= 1 + 1 / (2 \log(n+1))$. A close look at the proof reveals that
$\lim_{p \da 1} (p-1)^{1/p} C_1(p) = \exp(cw) \pi^{-1} (2 \kappa_2 + \kappa_3)$, and
for the M/G/1 example it is not difficult to show that
$\kappa_2 = \kappa_2(\delta) \leq 2 c^{-1} \delta^{-1}$, $\kappa_3 = \kappa_3(\delta) \leq (1 + \Var[X]) \delta^{-1} (1 + \rho)^{-1}$, $\Var[X] = \delta \lambda \E[B^2]$, and $\prob{A_n^{\mathfrak{c}}} = O(n^{-1})$. This means that, for large $n$, the right-handside of \eqref{eq:ubmg1} can be approximated by
\begin{align} \label{eq:anotherub}
(\alpha + \beta \delta^{-1}) e^{1/2} n^{-1/2} \log(n+1)
 \end{align}
 where
 \[ \alpha: = \kappa_4 + e^{cw} \pi^{-1} \lambda \E[B^2] (1 + \rho)^{-1},\:\:\:\:\: \beta := e^{cw} \pi^{-1} (4 c^{-1} + (1 + \rho)^{-1}). \]
If we neglect the $\log(n+1)$ term, then, on a fixed time horizon of length $T = \delta n$, the upper bound
\eqref{eq:anotherub} equals
\begin{align*}
(\alpha \delta^{1/2} + \beta \delta^{-1/2}) e^{1/2} T^{-1/2},
 \end{align*}
which suggests that $\delta$ should be chosen that minimizes $\alpha \delta^{1/2} + \beta \delta^{-1/2}$. In the application \cite{denBoerMandjesNunezZuraniewski2014} $\alpha$ and $\beta$ are unknown (because they depend on e.g.\ $\lambda$ and $\E[B^2]$), but if they can be replaced by known upper bounds $\alpha_u$ and $\beta_u$, then a heuristic choice for $\delta$ is to pick a minimizer of $\alpha_u \delta^{1/2} + \beta_u \delta^{-1/2}$ (yielding $\delta = \beta_u/\alpha_u$).
\end{remark}

\begin{remark} \label{rem:low}
Interestingly, the technique described above enables a fast and  accurate estimation of rare-event probabilities (i.e., $1-F^Y(w)$ for $w$ large), even in situations in which the estimation is based on input $X_1,\ldots,X_n$ for which the corresponding queue would not have exceeded level $w$. This idea, which resonates the concepts developed in \cite{MandjesVandeMeent2009}, has been worked out in detail in  \cite{denBoerMandjesNunezZuraniewski2014}. A numerical illustration of our estimator in this setting, and a comparison to the empirical estimator, is provided in Section \ref{sec:num}.
\end{remark}

\subsection{Decompounding} \label{subsec:decompounding}
Our second application involves decompounding a compound Poisson distribution, a concept that has been studied in the literature already (see the remarks on this in the introduction).

We start by providing a formal definition of the problem. Let $\mathcal{X}$ denote the collection of random variables of the form $\sum_{i=1}^{N'} Y'_i$, with $N'$ Poisson distributed with $\E[N'] > 0$, and  $(Y'_i)_{i \in \N}$ i.i.d.\ nonnegative random variables, independent of $N'$, and with $\prob{Y'_1 = 0} = 0$ (which can be assumed without loss of generality). For each $\tilde{X} \in \tilde{\mathcal{X}}$, let, for $s \in \C_{+}$,
\[ (\Psi \tilde{X})(s) = 1 + \frac{1}{- \log(\tilde{X}(\infty))} \Log( \tilde{X}(s)),\]
 where $\Log$ denotes the distinguished logarithm of $\tilde{X}$, and
 \[\tilde{X}(\infty):= \lim_{s \ra \infty, s \in \R} \tilde{X}(s) = \lim_{s \ra \infty, s \in \R} e^{\E[N] (-1 + \tilde{Y_1}(s))} = e^{-\E[N]} \]
 if $X = \sum_{i=1}^N Y_i$; here the last equality follows from $\prob{Y_1 = 0} = 0$.

Let $X = \sum_{i=1}^N Y_i$ be an element of $\mathcal{X}$, for some particular $Y \stackrel{\rm d}{=} Y_1$ and a Poisson distributed random variable $N$ with mean $\lambda > 0$. Since
 $-\log(\tilde{X}(\infty)) = \lambda$ and $\tilde{X}(s) = \exp(-\lambda + \lambda \E[-s Y])$, we have $\tilde{Y} = \Psi \tilde{X}$.
The idea is to estimate $F^Y(w)$, for $w > 0$, based on a sample $X_1, \ldots, X_n$ of $n \in \N$ independent copies of $X$, using the estimator $F_n^Y(w)$ of Section \ref{sec:method}, with, for $n\in{\mathbb N}$,
\[E_n: = \left\{ \frac{1}{n} \sum_{i=1}^n {\bf 1}_{\{X_i = 0\}} \in (0,1) \right\}\] and arbitrary $c> 0$.

\vspace{\baselineskip}
\begin{theorem} \label{thm:decompoundingrates}
Consider the estimation procedure outlined above.
Suppose 
$F^{Y}$ is continuously differentiable, twice differentiable in $w$, and suppose $\E[|X|^2] < \infty$. Then there is a constant $C>0$
such that
\[\E\left[|F_n^Y(w) - F^Y(w)|\right] \leq C n^{-1/2} \log(n+1) \]
for all $n \in \N$.
\end{theorem}
\begin{proof}
Write 
\[\lambda_n = -\log(\tilde{X}_n(\infty)) = -\log\left(\frac{1}{n} \sum_{i=1}^n {\bf 1}_{\{X_i = 0\}}\right)\] (being well-defined on $E_n$), and define
\[ A_{n,1} := \left\{ \sup_{-\sqrt{n} \leq y \leq \sqrt{n}}  |  \tilde{X}_n(c + \I y) -  \tilde{X}(c + \I y) | \leq \exp(-2 \lambda)/2 \right\}, \]
\[ A_{n,2} := \left\{ \frac{\lambda }{ 2} \leq \lambda_n \leq 2 \lambda \right\}, \]
and $A_n = A_{n,1} \cap A_{n,2}$. Note that $A_{n,2} \subset E_n$ and thus $A_n \subset E_n$. We show that assumptions (A1)--(A3) are valid. Because we explicitly assumed (A2), we are left with verifying   (A1) and (A3).
These verification resemble those of the M/G/1 example.

$\rhd$ {Assumption (A1).}
$\prob{A_{n,1}^{\mathfrak{c}}} = O( \sqrt{n} \exp(-n \beta^2 / 18)) = O(n^{-1/2})$ follows from Lemma \ref{lem:duvroye}, with $\beta = \exp(-2 \lambda)/2$, together with Chebyshev's Inequality and the assumption $\E[|X|^2] < \infty$.
$\prob{A_{n,2}^{\mathfrak{c}}} = O(n^{-1/2})$ follows from Hoeffding's Inequality, and thus
\[\prob{A_n^{\mathfrak{c}}} \leq \prob{A_{n,1}^{\mathfrak{c}}} + \prob{A_{n,2}^{\mathfrak{c}}} = O(n^{-1/2}).\]

$\rhd$ {Assumption (A3).} On $A_n$, for $s = c + \I y$, $-\sqrt{n} \leq y \leq \sqrt{n}$, we have
\[ \left|\,\frac{\tilde{X}_n(s) }{ \tilde{X}(s)} - 1\,\right| \leq \left|\tilde{X}_n(s) - \tilde{X}(s)\right|\, e^{2 \lambda} \leq \frac{1}{2},\]
where $|\tilde{X}(s)|^{-1} \leq \exp(2 \lambda)$ follows as in \eqref{eq:inftildeX}, and thus
\begin{eqnarray*}
\left|\Log(\tilde{X}_n(s))  - \Log(\tilde{X}(s)) \right|
&=& \left|\Log(\tilde{X}_n(s) / \tilde{X}(s)) \right|
= \left|L(\tilde{X}_n(s) / \tilde{X}(s)) \right| \\
&\leq&  \left|\tilde{X}_n(s) - \tilde{X}(s)\right| \,(\log 4) \,e^{2 \lambda},
\end{eqnarray*}
using \eqref{eq:Lfunction} and \eqref{eq:Lfunction2}.
This implies 
\begin{eqnarray*}
\lefteqn{\left| (\Psi \tilde{X}_n)(s) - (\Psi \tilde{X})(s) \right| \cdot {\bf 1}_{A_n} }\\
&\leq&
\left| \lambda_n^{-1} \Log(\tilde{X}_n(s)) - \lambda_n^{-1}  \Log(\tilde{X}(s))\right| \cdot {\bf 1}_{A_n}
+ \left| \lambda_n^{-1} \Log(\tilde{X}(s))   - \lambda^{-1}    \Log(\tilde{X}(s))\right| \cdot {\bf 1}_{A_n} \\
&\leq& \left|\lambda_n^{-1}\right| \cdot \left|\Log(\tilde{X}_n(s)) - \Log(\tilde{X}(s))\right| \cdot {\bf 1}_{A_n} +
\left| \lambda_n^{-1} - \lambda^{-1}\right| \cdot \left|\Log(\tilde{X}(s))\right| \cdot {\bf 1}_{A_n} \\
&\leq& \frac{2 \,\log 4}{\lambda} \,e^{2 \lambda} \cdot\left|\tilde{X}_n(s) - \tilde{X}(s)\right| \cdot {\bf 1}_{A_n} +
Z_n \:\text{ a.s.,}
\end{eqnarray*}
with $Z_n = 2 \lambda^{-2} |\lambda_n - \lambda| \cdot {\bf 1}_{A_n}$. By definition of $A_{n,2}$, $Z_n$ is bounded, and it follows from Hoeffding's inequality that there is a $\kappa_3 > 0$ independent of $n$ such that, for all $1 < p < 2$, $\E[|Z_n|^p] \leq \kappa_3 n^{-1/2}$.  This shows that (A3) is valid.
\end{proof}

\begin{remark} \label{rem:otherdistributions}
The decompounding example above can also be carried out with distributions other than Poisson. 
For example, if $N$ is Bin$(M, p)$ distributed, for known $M \in \mathbb{N}$ and unknown $p \in (0,1)$, then
$\tilde{X}(s) = ( p \tilde{Y}(s) + 1-p)^M$, $\tilde{X}(\infty) = (1-p)^M$, and thus
\[\tilde{Y}(s) = (\Psi \tilde{X})(s) := \frac{\tilde{X}(s)^{1/M} - \tilde{X}(\infty)^{1/M}}{1 - \tilde{X}(\infty)^{1/M}}.\]
Or, if $N$ is negative binomially distributed, i.e.
\[ \prob{N=n} = \left( \begin{array}{cc} n + M - 1 \\ n \end{array} \right) (1-p)^M p^n, \quad (n = 0,1,2,\ldots), \]
for some known $M \in \N$ and unknown $p \in (0,1)$, then 
$\tilde{X}(s) = (1-p)^M (1 - p \tilde{Y}(s))^{-M}$, 
$\tilde{X}(\infty) = (1-p)^M$, and thus
\[ \tilde{Y}(s) = (\Psi \tilde{X})(s) := \frac{1 - \tilde{X}(\infty)^{1/M} \tilde{X}(s)^{-1/M}}{1 - \tilde{X}(\infty)^{1/M}}. \]
For both examples it is not difficult to construct $A_n$ and $Z_n$, in the same spirit as in the proof of Theorem \ref{thm:decompoundingrates}, such that the convergence rates $\mathbb{E}[|F_n^Y(w) - F^Y(w)|] =  O(n^{-1/2} \log(n+1))$ hold. The key requirement on $N$ to obtain these rates is that the relation $\tilde{X}(s) = \mathbb{E}[ \tilde{Y}(s)^N]$ can be inverted, such that we can write $\tilde{Y}(s) = (\Psi \tilde{X})(s)$ for some mapping $\Psi$.
\end{remark}

\section{Auxiliary lemmas} \label{sec:aux}
This section contains a number of auxiliary lemmas that are used in the proofs of Theorems \ref{thm:convergenceRate}, \ref{thm:mg1rates}, and~\ref{thm:decompoundingrates}.

\begin{lemma} \label{lem:sqrtlemma}
Let $c > 0$, $n \in \N$ and let $X_1, \ldots, X_n$ be i.i.d.\ nonnegative random variables distributed as $X$. For all $p \in (1,2)$ and $s \in c + \I \R$,
\begin{eqnarray*}
 \E\left[ \left| \tilde{X}(s) - \frac{1}{n} \sum_{i=1}^n \exp(-s X_i) \right|^p \right] \leq 2^p n^{-1/2},
 \end{eqnarray*}
 and
\begin{eqnarray*}
 \E\left[ \left| \E[X] - \frac{1}{n} \sum_{i=1}^n X_i \right|^p \right] \leq  (1 + \Var[X]) n^{-1/2},
 \end{eqnarray*}
  where the last inequality is only informative if \,$\Var[X]<\infty$.
 \end{lemma}
\begin{proof}
Let $s \in c + i \R$. Since $X_i \geq 0$ a.s.\ for all $i=1, \ldots, n$, we have
\begin{align*}
\left| \tilde{X}(s) - \frac{1}{n} \sum_{i=1}^n \exp(-s X_i) \right|^{p-1} \leq
\left( |\tilde{X}(s)| + \frac{1}{n} \sum_{i=1}^n \left| \exp(-s X_i) \right|  \right)^{p-1} \leq 2^{p-1} \text{ a.s.}
\end{align*}
Jensen's Inequality then implies
\begin{align*}
\E\left[ \left| \tilde{X}(s) - \frac{1}{n} \sum_{i=1}^n \exp(-s X_i) \right|^p \right]
&\leq 2^{p-1} \sqrt{ \E\left[ \left|  \frac{1}{n} \sum_{i=1}^n (\exp(-s X_i) - \E[\exp(-s X)]) \right|^2 \right] } \\
&= 2^{p-1} \sqrt{ \frac{1}{n} \E[\left| \exp(-s X) - \E[\exp(-s X)] \right|^2]}
\leq 2^p n^{-1/2}.
\end{align*}
Furthermore, we have, again by Jensen's Inequality,
\begin{eqnarray*}
\lefteqn{ \E\left[ \left| \E[X] - \frac{1}{n} \sum_{i=1}^n X_i \right|^p \right]
\leq n^{-p} \E\left[ \left| \sum_{i=1}^n (X_i - \E[X]) \right|^2 \right]^{p/2}}\\
&\leq& n^{-p} n^{p/2} \E[(X - \E[X])^2]^{p/2} \leq n^{-1/2}  \Var[X]^{p/2}
\leq n^{-1/2} (1 + \Var[X]).
\end{eqnarray*}
This proves the claims.
\end{proof}

\begin{lemma} \label{lem:duvroye}
Let $n \in \N$, and let $X_1, \ldots, X_n$ be i.i.d.\ nonnegative random variables distributed as $X$.
Let $\alpha > 0$, $\beta > 0$, $c > 0$, and \[\tilde{X}_n(s) := \frac{1}{n} \sum_{i=1}^n \exp(-s X_i),\] for $s \in \C_{+}$.
Then
\begin{eqnarray*}
\prob{ \sup_{|t| \leq \alpha} | \tilde{X}(c + \I t) - \tilde{X}_n(c + \I t)| > \beta} &<& 4 \left(1 + \frac{8 \alpha \E[|X|]}{\beta} \right) \exp(-n \beta^2 / 18) \\
&& +\: \prob{ \left| \frac{1}{n} \sum_{i=1}^n X_i \right| \geq \frac{4}{3} \E[|X|]}.
\end{eqnarray*}
\end{lemma}
\begin{proof}
One can show that, for all $t,s \in [-\alpha, \alpha]$,
\[ |\tilde{X}(c + \I t) - \tilde{X}(c + \I s)| \leq \E\left[|1 - \exp(\I (t-s) X)|\right],\]
and
\[ |\tilde{X}_n(c + \I t) - \tilde{X}_n(c + \I s)| 
\leq |t-s| \left| \frac{1}{n} \sum_{i=1}^n X_i \right|,\]
whereas, for each $t_i \in [-\alpha, \alpha]$,
\begin{eqnarray*}
\lefteqn{\prob{ |\tilde{X}(c + \I t_i) - \tilde{X}(c + \I t_i)| > \frac{1}{3} \beta}} \\
&\leq &\prob{ | \Re\big(\tilde{X}(c + \I t_i) - \tilde{X}(c + \I t_i)\big)| > \frac{1}{6} \beta} +
\prob{ | \Im\big(\tilde{X}(c + \I t_i) - \tilde{X}(c + \I t_i)\big)| > \frac{1}{6} \beta}\\
&\leq &4 \exp(-2 n \beta^2 / 36),
\end{eqnarray*}
using Hoeffding's inequality. 
The  claim then follows along precisely the same lines as the proof of \citep[Theorem 1]{Devroye1994}.
\end{proof}

\begin{lemma} \label{lem:ntaillemma}
Let $w > 0$, $c > 0$, and let $f:[0, \infty) \ra [0,1]$ be a continuously differentiable function, twice differentiable in the point $w$, and such that $\int_w^{\infty} |f'(y+w)| e^{-cy} y^{-1} \,\dy < \infty$. 
There exists a $\kappa_4 > 0$ such that, for all $m > 0$,
\[ \left| \int_{|y|>m} \frac{1}{2 \pi} e^{(c+\I y) w} \bar{f}(c+\I y) \DD y
\right|  \leq \frac{\kappa_4}{m}. \]
\end{lemma}
\begin{proof}
Fix $m > 0$. Observe that
\begin{eqnarray} \nonumber
\lefteqn{ \int_{|y|\leq m} \frac{1}{2 \pi} e^{(c+\I y) w} \bar{f}(c+\I y) \dy
=   \int_{|y|\leq m} \frac{1}{2 \pi} e^{(c+\I y) w} \int_0^{\infty} e^{-(c+\I y) x} f(x) \dx\,\dy } \\ \nonumber
&=&   \int_0^{\infty} \frac{1}{2 \pi} f(x) \int_{|y|\leq m} e^{(c+\I y) (w-x)} \dy\,\dx
=   \int_0^{\infty} \frac{1}{\pi} f(x) e^{c (w-x)} \frac{\sin(m (w-x))}{w-x} \dx  \\
&=&  \int_{-w}^{\infty} \frac{1}{\pi} f(y + w) e^{-c y} \frac{\sin(m y)}{y} \dy,  \label{eq:mstepeq1}
\end{eqnarray}
using Fubini's Theorem and the variable substitution $y: = x-w$, together with the obvious identity $\sin(-m y) /( -y) = \sin(m y)/y$.

We consider the integral \eqref{eq:mstepeq1} separately over the domain $[w, \infty)$ and $[-w,w]$.
For the interval $[w, \infty)$, we have
\begin{eqnarray}  \nonumber
\lefteqn{ \left|\int_w^{\infty} \frac{1}{\pi} \frac{f(y + w) e^{-c y}}{y} \sin(m y) \dy \right|}\\
&\leq&\left| \left[  \frac{1}{\pi} \frac{f(y + w) e^{-c y}}{y} \frac{\cos(m y)}{-m} \right]_{y=w}^{\infty} \right|
+ \left|\int_w^{\infty} \frac{1}{\pi} \pd{}{y} \left[ \frac{f(y + w) e^{-c y}}{y} \right] \frac{\cos(m y)}{m} \dy \right|\nonumber  \\ \nonumber
&\leq & \left| \frac{1}{\pi} \frac{f(w + w) e^{-c w}}{w} \frac{\cos(m w)}{m} \right|\nonumber \\
&&+\: \int_w^{\infty} \frac{1}{\pi} |f'(y+w)| e^{-cy} y^{-1} \frac{1}{m} \dy
+ \int_w^{\infty} \frac{1}{\pi} f(y+w) (c e^{-cy} y^{-1} + e^{-c y } y^{-2}) \frac{1}{m} \dy\nonumber \\ \label{eq:prooflemmaeq0}
&\leq & \frac{1}{m} \cdot \left(\frac{e^{-c w}}{w \pi} +
\frac{1}{\pi} \int_w^{\infty} |f'(y+w)| e^{-cy} y^{-1} \dy  + \frac{e^{-c w}}{\pi w} + \frac{e^{-c w}}{\pi c w^2} \right).
\end{eqnarray}
We now consider the integral \eqref{eq:mstepeq1} on the interval $[-w, w]$.
Write $\phi(y) := f(y + w) e^{-c y}$ and $g(y) := (\phi(y) - \phi(0) - \phi'(0) y)/y$, and observe that $g$ is continuously differentiable on the interval $[-w, w]$ (which follows from the fact that $f''(w)$ exists). We have
\begin{eqnarray} \nonumber
\lefteqn{\left|\phi(0) - \int_{-w}^{w} \frac{1}{\pi} f(y + w) e^{-c y} \frac{\sin(m y)}{y} \dy \right|} \\ \nonumber
&=& \left|\phi(0) - \int_{-w}^{w} \frac{1}{\pi} \Big(\phi(0) + \phi'(0) y + g(y) y \Big) \frac{\sin(m y)}{y} \dy \right| \\ \label{eq:prooflemmaeq1}
&\leq & \phi(0) \left|1 - \int_{-w}^{w} \frac{1}{\pi} \frac{\sin(m y)}{y} \dy \right|
+ \left|\int_{-w}^{w} \frac{1}{\pi} g(y) \sin(m y) \dy \right|;
\end{eqnarray}
realize that $\int_{-w}^w \pi^{-1} \phi'(0) \sin(my) {\rm d}y =0$.

We first bound the first term of \eqref{eq:prooflemmaeq1}.
\begin{align*}
\left|1 - \int_{-w}^{w} \frac{\sin(m y)}{\pi y} \dy \right|
\leq \left|1 - \int_{-\infty}^{\infty}  \frac{\sin(m y)}{\pi y} \dy \right|
+ \left|  \int_{w}^{\infty}   \frac{2\sin(m y)}{\pi y} \dy \right|
=  \left| \int_{w}^{\infty} \frac{2\sin(m y)}{\pi y} \dy \right|.
\end{align*}
Write $h(a) := \int_w^{\infty} e^{-a y} \,y^{-1}\,{\sin(m y)} \dy$, $a \geq 0$.
Then $\lim_{a \ra \infty} h(a) = 0$,
\begin{eqnarray*}
h'(a) &=& \int_w^{\infty} -e^{-a y} \sin(m y) \dy \\&=& -e^{-a w} \int_0^{\infty} e^{-a x} \sin(m(x+w)) \dx
= -e^{-a w} \frac{m \cos(w m) + a \sin(w m)}{a^2+m^2},
\end{eqnarray*}
and thus
\begin{eqnarray*}
\left| \int_{w}^{\infty} \frac{\sin(m y)}{y} \dy \right| & = &|h(0) |  = \left| \lim_{a \ra \infty} h(a) - \int_0^{\infty} h'(a) {\rm d}a \right|\\
&=& \left| \int_0^{\infty} e^{-a w} \frac{m \cos(w m) + a \sin(w m)}{a^2+m^2}{\rm d}a \right| \\
&\leq&  \int_0^{\infty} e^{-a w} \frac{m + a}{a^2+m^2} {\rm d}a
\leq  \frac{2}{m} \int_0^{\infty} e^{-a w}  {\rm d}a = \frac{2}{m w},
\end{eqnarray*}
which implies
\begin{align} \label{eq:prooflemmaeq2}
\left|1 - \int_{-w}^{w}  \frac{\sin(m y)}{\pi y} \dy \right|
\leq \frac{4}{w \pi m}.
\end{align}

The second term of \eqref{eq:prooflemmaeq1} is bounded by
\begin{eqnarray} \nonumber
 \lefteqn{\left|\int_{-w}^{w} \frac{1}{\pi} g(y) \sin(m y) \dy \right|} \\ \nonumber
&=&
\left|\frac{1}{\pi} g(w) \frac{\cos(-m w)}{m} - \frac{1}{\pi} g(-w) \frac{\cos(m w)}{m}
- \int_{-w}^{w} \frac{1}{\pi} g'(y) \frac{\cos(-m y)}{m} \dy \right| \\ \label{eq:prooflemmaeq3}
&\leq& \frac{|g(w) - g(-w)|}{\pi m} + \frac{1}{\pi m} \int_{-w}^w |g'(y)| \dy.
\end{eqnarray}

Combining \eqref{eq:mstepeq1}, \eqref{eq:prooflemmaeq0}, \eqref{eq:prooflemmaeq1}, \eqref{eq:prooflemmaeq2} and \eqref{eq:prooflemmaeq3}, using $f(w) = \phi(0)$, it follows that
\begin{eqnarray*}
\lefteqn{\left| \int_{|y| > m} \frac{1}{2 \pi} e^{(c+\I y) w} \overline{f}(c+\I y) \dy \right|}\\
&=&\left| f(w) - \int_{|y|\leq m} \frac{1}{2 \pi} e^{(c+\I y) w} \overline{f}(c+\I y) \dy \right|\\
&=&  \left| f(w) - \int_{-w}^{w} \frac{1}{\pi} f(y + w) e^{-c y} \frac{\sin(m y)}{y} \dy
- \int_{w}^{\infty} \frac{1}{\pi} f(y + w) e^{-c y} \frac{\sin(m y)}{y} \dy \right| \\
&\leq & f(w)   \frac{4}{\pi m w} +
\frac{|g(w) - g(-w)|}{\pi m} + \frac{1}{\pi m} \int_{-w}^w |g'(y)| \dy \\
 &&+\,\frac{1}{m} \cdot \left(\frac{e^{-c w}}{w \pi} +
\frac{1}{\pi} \int_w^{\infty} |f'(y+w)| e^{-cy}{y} \dy  + \frac{e^{-c w}}{\pi w} + \frac{e^{-c w}}{\pi c w^2} \right).
\end{eqnarray*}
Defining
\begin{eqnarray*}
\kappa_4 &:=& f(w)   \frac{4}{\pi w} +
\frac{|g(w) - g(-w)|}{\pi} + \frac{1}{\pi} \int_{-w}^w |g'(y)| \dy \\
& &+\: \frac{e^{-c w}}{w \pi} +
\frac{1}{\pi} \int_w^{\infty} |f'(y+w)| e^{-cy}{y} \dy  + \frac{e^{-c w}}{\pi w} + \frac{e^{-c w}}{\pi c w^2},
\end{eqnarray*}
this implies the stated of the lemma.
\end{proof}

\section{Numerical illustration} \label{sec:num}
We provide a numerical illustration of the performance of our estimator, inspired by an application of estimating high-load probabilities in communication links \cite{denBoerMandjesNunezZuraniewski2014}. In particular, we consider an M/G/1 queue in stationarity that serves jobs at unit speed, and whose (unknown) service time distribution is exponential with mean $1/20$. We choose the (unknown) arrival rate $\lambda$ from $\{10, 18, 19\}$; this corresponds to load factors $\rho$ of 0.50, 0.90, and 0.95. 
For $n = 10,000$ consecutive time intervals of length $\delta = 0.10$, the amount of work arriving to the queue in each interval is recorded.
Based on these samples, we estimate the tail probabilities $\prob{Y>w}$ of the workload distribution $Y$ for different values of $w$, using the Laplace-transform based estimator outlined in Section \ref{subsec:mg1}. We test values of $w$ corresponding to the 90th, 99th, and 99.9th percentile of $Y$; the particular values, denoted by $w_{.9}$, $w_{.99}$, and $w_{.999}$, are given in Table \ref{table:w}. 

\begin{table}[!ht]
\begin{center}
\caption{90th, 99th, and 99.9th percentiles of $W$, for different values of $\rho$.}
\label{table:w}
\begin{tabular}{r|rrr}
$\rho$ & $w_{.9}$ & $w_{.99}$ & $w_{.999}$ \\ \hline
0.50 & 0.1609 & 0.3912	& 0.6215 \\
0.90 & 1.0986 & 2.2499  & 3.4012 \\
0.95 & 2.2513	& 4.5539  & 6.8565 
\end{tabular}
\end{center}
\end{table}

For each $\rho \in \{0.50, 0.90, 0.95\}$ and each of the three corresponding values of $w$, we run 1000 simulations and record the relative estimation error
\begin{equation} \label{eq:relativererror}
 \left| \frac{ (1 - F_n^Y(w)) - \prob{Y>w}}{\prob{Y>w}} \right|, 
 \end{equation}
where $F_n^Y(w)$ denotes the outcome of the Laplace-transform based estimator. The simulation average of \eqref{eq:relativererror}, for different values of $\rho$ and $w$, is reported in Table \ref{table:outcomes}, at the lines starting with `Laplace'.

We compare the performance of the Laplace-transform based estimator to that of the empirical estimator that samples the workload $W(i \delta)$ at time points $i \delta$, $i=1,\ldots, n$, and estimates the tail probability $\prob{Y>w}$ by the fraction $n^{-1} \sum_{i=1}^n {\bf 1}_{Y(i \delta) > w}$. The corresponding simulation average of the relative estimation error is reported in Table \ref{table:outcomes}, at the lines starting with `Empirical'.

Table \ref{table:outcomes} shows that the Laplace-transform based estimator has a lower relative error than the empirical estimator, for all-but-one tested instances of $\rho$ and $w$.
This is perhaps not surprising, since the `Laplace' estimator is based on i.i.d.\ samples (of the amount of work arriving to the queue in $\delta$ time units), whereas the `Empirical' estimator is based on correlated samples (of the workload in the queue). 

A third estimator, that is based on the same samples as the `Empirical' estimator, can be constructed as follows: consider the samples of the workload process $Y(i \delta)$, $i=1,\ldots, n$, and let $Q = \{ Y(i \delta) - (Y((i-1) \delta) - \delta) \mid Y((i-1) \delta) \geq \delta, 2 \leq i \leq n\}$. 
If, for some $i$, $Y((i-1) \delta) \geq \delta$, then the amount of work arrived in the $\delta$ time units prior to time point $i \delta$ is precisely equal to
$Y(i \delta) - (Y((i-1) \delta) - \delta)$. (If $Y((i-1) \delta) < \delta$, then the exact amount of work arrived between time points $(i-1)\delta$ and $i \delta$ can not be inferred from the workload samples). If we apply the Laplace-transform based estimator on the samples in the set $Q$ (which are independent samples from the amount of work arriving to the queue in $\delta$ time units), then we obtain an estimate of $\prob{Y>w}$ that is based on the same samples as the `Empirical' estimator. The relative estimation error of this third estimator is reported in Table \ref{table:outcomes}, at the lines starting with `Laplace, censored'.

\begin{table}[ht]
\caption{Average relative estimation error}
\label{table:outcomes} 
\begin{tabular}{lr} \\
$\rho = 0.50$ \, \, \, \, \, &
\begin{tabular}{r|rrr}
Estimator & $w = w_{.9}$ & $w = w_{.99}$ & $w = w_{.999}$ \\ \hline
Laplace	&	0.05 &	0.13	 & 0.25\\
Empirical	&	0.50 &	0.50&	0.67\\
Laplace, censored	&	0.15&	0.39	&0.67
\end{tabular} \\
 & \\
$\rho = 0.90$ \, \, \, \, \, &
\begin{tabular}{r|rrr}
Estimator & $w = w_{.9}$ & $w = w_{.99}$ & $w = w_{.999}$ \\ \hline
Laplace	&	0.19 &		0.40 & 0.65 \\
Empirical	&	0.29  &	0.96 &	1.82\\
Laplace, censored	& 0.23	&	0.49	& 0.81
\end{tabular} \\
& \\
$\rho = 0.95$ \, \, \, \, \, &
\begin{tabular}{r|rrr}
Estimator & $w = w_{.9}$ & $w = w_{.99}$ & $w = w_{.999}$ \\ \hline
Laplace	&	0.39  &  0.96 &   2.09 \\
Empirical	&	 0.52  &  1.36 &   1.83	\\
Laplace, censored	&	0.43 & 1.07 & 2.34
\end{tabular} 
\end{tabular}
\end{table}

Table \ref{table:outcomes} shows that the `Laplace, censored' estimator still outperforms the `Empirical' estimator, in all-but-one instances. Both these estimators are based on the same samples of the workload process. A notable disadvantage of `Empirical' estimator is that it requires the system to reach high load in order to obtain informative estimates. In practice, particularly in the context of operated communication links, this is not desirable: network operators would certainly intervene if the network load reaches exceedingly high levels. These interventions hamper the estimation of the probability that this high load occurs. In contrast, both the Laplace-transform based estimators produce informative estimates of $\prob{Y>w}$, even if all sampled values of the workload process are below $w$.

\section{Discussion, concluding remarks} \label{sec:disc}
In this paper we have discussed a technique to estimate the distribution of a random variable $Y$, focusing on the specific context in which we have i.i.d.\ observations $X_1,\ldots,X_n$, distributed as a random variable $X$, where the relation between the Laplace transforms of $X$ and $Y$ is known. Our problem was motivated from a practical question of an internet service provider, who wished to develop statistically sound techniques to estimate the packet delay distribution based on various types of probe measurements; specific quantiles of the delay distribution are mutually agreed upon by the service provider and its customers, and posted in the service level agreement.
To infer whether these service level agreements are met, the internet provider estimates several tail probabilities of the delay distribution.
This explains why we have focused on the setup presented in our paper, concentrating on estimating the distribution function $F^Y(w)$ and bounding the error ${\mathbb E}[|F_n^Y(w)-F^Y(w)|]$ for this $w$. It is noted that various other papers focus on estimating the density, and often use different convergence metrics; some establish asymptotic Normality.

A salient feature of our analysis is that the ill-posedness of Laplace inversion, i.e., the fact that the inverse Laplace transform operator is not continuous, does not play a r\^ole. Our estimate $F_n^Y(w)$ is `close' to $F^Y(w)$ if the Laplace transform $\bar{F}_n^Y$ is `close' to the Laplace transform $\bar{F}^Y$, measuring `closeness' of these Laplace transforms by the integral \eqref{eq:term3inproof}. Our assumptions (A1)-(A3) ensure that this integral converges to zero (as $n$ grows large), and Section \ref{sec:appl} shows that these conditions are met in practical applications. We therefore do not need regularized inversion techniques as in \cite{MnatsakanovRuymgaartRuymgaart2008} and \cite{Shimizu2010}, with convergence rates of just $1 / \log(n)$. (See further Remark \ref{rem:cdf}).

\vspace{2mm}

{\sc Acknowledgments ---}
{\small
This research is partially funded by SURFnet, Radboudkwartier 273,
3511 CK Utrecht,
The Netherlands. We thank Rudesindo N\'u\~nez-Queija (University of Amsterdam) and Guido Janssen (Eindhoven University of Technology, the Netherlands) for useful discussions and providing literature references. The constructive comments and suggestions of the anonymous referees have improved the paper, and are kindly acknowledged. Part of this work was done while the first author was affiliated with Eindhoven University of Technology and University of Amsterdam.}

\end{document}